\documentclass[letter,11pt]{article}

\usepackage{fullpage}
\usepackage{pgfplots}
\usepackage{epsfig, amssymb, amsmath}
\usepackage{amsfonts}
\usepackage{graphicx,color}
\usepackage{bm}

\usepackage{multirow}

\allowdisplaybreaks

\def\RP{{\mathbb{R}_{\geq 0}}}
\def\RPP{{\mathbb{R}_{>0}}}
\def\R{{\mathbb{R}}}

\def\N{[n]}
\def\S{\beta\mbox{-}{\sf SUM}}
\def\M{\beta\mbox{-}{\sf MAX}}

\newtheorem{theorem}{Theorem}

\newtheorem{lemma}{Lemma}
\newtheorem{definition}{Definition}
\newtheorem{property}{Property}

\def \qed{{$\hfill \Box$}}

\newenvironment{proof}{{\bf Proof}: }

\begin{document}

\title{On the Robustness of the Approximate Price of Anarchy\\in Generalized Congestion Games\\(Full Version)\thanks{This work was partially supported by the PRIN 2010--2011 research project ARS TechnoMedia: ``Algorithmics for Social Technological Networks'' funded by the Italian Ministry of University.}}
\author{
Vittorio Bil\`o \thanks{{Department of Mathematics and Physics ``Ennio De Giorgi", University of Salento,
Provinciale Lecce-Arnesano, P.O. Box 193, 73100 Lecce - Italy,
Email: \textsf{vittorio.bilo@unisalento.it}.}}}

\maketitle

\begin{abstract}
One of the main results shown through Roughgarden's notions of smooth games and robust price of anarchy is that, for any sum-bounded utilitarian social function, the worst-case price of anarchy of coarse correlated equilibria coincides with that of pure Nash equilibria in the class of weighted congestion games with non-negative and non-decreasing latency functions and that such a value can always be derived through the, so called, smoothness argument. We significantly extend this result by proving that, for a variety of (even non-sum-bounded) utilitarian and egalitarian social functions and for a broad generalization of the class of weighted congestion games with non-negative (and possibly decreasing) latency functions, the worst-case price of anarchy of $\epsilon$-approximate coarse correlated equilibria still coincides with that of $\epsilon$-approximate pure Nash equilibria, for any $\epsilon\geq 0$. As a byproduct of our proof, it also follows that such a value can always be determined by making use of the primal-dual method we introduced in a previous work. It is important to note that our scenario of investigation is beyond the scope of application of the robust price of anarchy (for as it is currently defined), so that our result seems unlikely to be alternatively proved via the smoothness framework.
\end{abstract}

\section{Introduction}
The celebrated notion of robust price of anarchy introduced by Roughgarden in \cite{R09,R12} has lately arouse much interest in the determination of inefficiency bounds for pure Nash equilibria which may automatically extend to some of their appealing generalizations, such as mixed Nash equilibria, correlated equilibria and coarse correlated equilibria. These three types of solutions have a particular flavor since, differently from pure Nash equilibria, they are always guaranteed to exist by Nash's Theorem \cite{N50}\footnote{To this aim, we recall that the set of coarse correlated equilibria contains that of correlated equilibria, which contains that of mixed Nash equilibria, which contains that of pure Nash equilibria.}; moreover, the last two ones can also be efficiently computed and even easily learned when a game is repeatedly played over time.

To this aim, Roughgarden \cite{R09,R12} identifies a class of games, called {\em smooth games}, for which a simple three-line proof, called {\em smoothness argument}, shows significant upper bounds on the price of anarchy of pure Nash equilibria as long as the social function measuring the quality of any strategy profile in the game is {\bf\em sum-bounded}, that is, upper bounded by the sum of the players' costs\footnote{Throughout the paper, we implicitly assume that all games under consideration are cost minimization ones. All the claimed properties and results can be applied {\em mutatis mutandis} to the case of payoff maximization games.}. He then defines the {\em robust price of anarchy} of a smooth game as the best-possible (i.e., the lowest) upper bound which can be derived by making use of this argument and provides an {\em extension theorem} which shows that, still for sum-bounded social functions,
the price of anarchy of coarse correlated equilibria of any smooth game is upper bounded by its robust price of anarchy.
Finally, he shows that several games considered in the literature happen to be smooth and that the class of (unweighted) congestion games with non-negative and {\bf\em non-decreasing latency functions} is {\em tight} for the {\bf\em utilitarian social function} (that is, the social function defined as the sum of the players' costs), in the sense that, in this class of games, the worst-case price of anarchy of pure Nash equilibria exactly matches the robust price of anarchy. This last result has been subsequently extended to the class of weighted congestion games by Bhawalkar, Gairing and Roughgarden in \cite{BGR10}.

\subsection{Our Contribution and Significance}
In this work, we generalize the tightness result by Bhawalkar, Gairing and Roughgarden along the following four directions (see Section \ref{sec-def} for formal definitions):
\begin{enumerate}
\item the class of games we consider is a broad generalization of that of weighted congestion games. In particular, we focus on {\em generalized weighted congestion games}, that is, games in which each player's {\em perceived cost} is defined as a certain linear combination of all the players' {\em individual costs} originally experienced in some underlying weighted congestion game. Thus, it is quite easy to figure out that the class of generalized weighted congestion games widely extends that of weighted congestion games;
\item the families of social functions we consider are generalizations of both the utilitarian and the egalitarian social functions (where the egalitarian social function is defined as the maximum of the players' costs). In particular, a family of utilitarian social functions is obtained by summing up a certain contribution from each player, whereas a family of egalitarian social functions is obtained by taking the maximum contribution among the players, where each player's contribution is given by a conic combination of the players' individual costs. We stress that such a combination may significantly differ from the one used to define the players' perceived costs, so that there exist social functions in both families that may not be sum-bounded;
\item the latency functions we consider in the definition of the players' individual costs are selected from a family of allowable non-negative functions with no additional restrictions. This permits us to encompass also latency functions not considered so far in the previous tightness results known in the literature, such as, for instance, the widely used fair cost sharing rule induced by the Shapley value \cite{S53};
\item the solution concepts we consider are the approximate versions of all the four types of equilibria named so far. In particular, for any real value $\epsilon\geq 0$, we focus on either $\epsilon$-approximate pure Nash equilibria and $\epsilon$-approximate coarse correlated equilibria. For the special case of $\epsilon=0$, one reobtains the notions of pure Nash equilibria and coarse correlated equilibria, so that results for these solution concepts can be obtained as a special case of the ones holding for their approximate versions.
\end{enumerate}
More precisely, but still informally speaking, we prove the following result (Theorem \ref{mainth} in Section~\ref{sec-main}):
\begin{quote}
{\em for a variety of utilitarian and egalitarian social functions and for any real value $\epsilon\geq 0$, the worst-case price of anarchy of $\epsilon$-approximate pure Nash equilibria coincides with that of $\epsilon$-approximate coarse correlated equilibria in the class generalized weighted congestion games with non-negative latency functions}.
\end{quote}
As it can be appreciated, the above tightness result generalizes the previous one by Bhawalkar, Gairing and Roughgarden along all four directions simultaneously. The technique we use to prove the theorem is the primal-dual method that we introduced in \cite{B12}. In fact, as a byproduct of our proof, it also follows that, in the above considered scenario of investigation,
\begin{quote}
{\em the worst-case price of anarchy of $\epsilon$-approximate pure Nash equilibria can always be determined through the primal-dual method}.
\end{quote}
We would like to stress that, when adopting the social functions described at point $2$, generalized weighted congestion games are not smooth games in general, so that the above tightness result seems unlikely to be reproved via smoothness arguments, at least in the way in which they have been defined and used so far in the literature. This seems to provide an evidence that the primal-dual method may be more powerful than the smoothness framework as far as we focus on congestion games and some of their possible generalizations.

\subsection{Related Work}
The notion of price of anarchy as a measure of the inefficiency caused by selfish behavior in non-cooperative games has been introduced in a seminal paper by Koutsoupias and Papadimitriou \cite{KP99} in 1999. Since then, several classes of games have been studied under this perspective. Among these classes, congestion games introduced by Rosenthal in \cite{R73} and their weighted variants \cite{MS96} occupy a preeminent role.

Awerbuch, Azar and Epstein \cite{AAE05} and Christodoulou and Koutsoupias \cite{CK05} focus on the worst-case price of anarchy of pure Nash equilibria in either weighted and unweighted congestion games under the utilitarian social function. They independently give tight bounds for the case of affine latency functions and almost tight upper and lower bounds for the case of polynomial latency functions with non-negative coefficients. Such a gap has been subsequently closed by Aland et al. in \cite{ADGMS11}. Moreover, Christodoulou, Koutsoupias and Spirakis \cite{CKS11} obtain tight bounds on the worst-case price of anarchy of approximate pure Nash equilibria in unweighted congestion games for the case of polynomial latency functions with non-negative coefficients, while Christodoulou and Koutsoupias \cite{CK05b} show that the worst-case price of anarchy of correlated equilibria is the same as that for pure Nash equilibria in weighted and unweighted congestion games when considering affine latency functions. As already said, such an equivalence has been further extended to coarse correlated equilibria and to any class of non-negative and non-decreasing latency functions by Roughgarden \cite{R09,R12} in the unweighted case and by Bhawalkar, Gairing and Roughgarden \cite{BGR10} in the weighted case, by making use of the smoothness argument and the robust price of anarchy.

Robust bounds on the worst-case price of anarchy have been lately achieved via extensions of the smoothness argument in some generalizations of (unweighted) congestion games. In particular, de Keijzer et al. \cite{DSAB13} and Rahn and Sch\"afer \cite{RS13} consider the altruistic extension of congestion games in which, similarly to our model of generalized congestion games, the perceived cost of each player is defined as a linear combination of the individual costs of all the players in the game. Anyway, while we do not impose any kind of restriction on such a combination, they consider the case in which the multiplicative coefficients lie in the interval $[0,1]$ and, for each player $i$, the contribution of the individual cost of player $i$ to her perceived cost has to be always multiplied by $1$. Moreover, they restrict their analysis to the case in which the social function is the sum of the players' individual costs.

Much less attention has been devoted in the literature to the egalitarian social function, for which Christodoulou and Koutsoupias \cite{CK05} give an asymptotically tight bound on the worst-case price of anarchy in unweighted congestion games with affine latency functions.

We introduced the primal-dual method in \cite{B12} as a tool for obtaining tight bounds on the inefficiencies caused by selfish behavior in weighted congestion games and their possible generalizations for a variety of solutions concepts. In particular the primal-dual method has been applied by Bil\`o, Flammini and Gallotti \cite{BFG12} to derive tight bounds on the worst-case price of anarchy of pure Nash equilibria in congestion games with affine latency functions under the assumption that the players' knowledge is restricted by the presence of an underlying social knowledge graph; by Bil\`o \cite{B14} to derive tight bounds on the worst-case price of stability of pure Nash equilibria in congestion games with affine latency functions and altruistic players; by Bil\`o and Paladini \cite{BP14} to derive tight bounds on the approximation ratio of the solutions achieved after a one-round walk of $\epsilon$-approximate best-responses starting from any initial strategy profile in cut games, for any $\epsilon\geq 0$; by Bil\`o et al. \cite{BFMM14} to derive a surprising matching lower bound on the price of anarchy of subgame perfect equilibria in sequential cut games; and by Bil\`o, Fanelli and Moscardelli \cite{BFM13} to derive significant upper bounds on the price of anarchy of lookahead equilibria in congestion games with affine latency functions.

\subsection{Paper Organization}
The paper is organized as follows. In the next section, we give all necessary definitions and notation and provide also some preliminary remarks. Section \ref{sec-main} contains the technical contribution of the paper, with the proof of our main theorem. In the last section, we conclude and discuss open problems.

\section{Definitions, Notation and Preliminaries}\label{sec-def}
A {\em weighted congestion game} is a tuple ${\sf CG}=\left(\N,(w_i)_{i\in\N},E,({\sf\Sigma}_i)_{i\in\N},(\ell_e)_{e\in E}\right)$ such that $\N=\{1,2,\ldots,n\}$ is a set of $n\geq 2$ players, $w_i>0$ is the {\em weight} of player $i$, $E$ is a non-empty set of {\em resources}, ${\sf\Sigma}_i\subseteq 2^E\setminus\{\emptyset\}$ is a non-empty {\em set of strategies} for player $i$ and $\ell_e:\RP\rightarrow\RP$ is the {\em latency function} of resource $e\in E$. Denote as ${\sf\Sigma}=\prod_{i\in\N}{\sf\Sigma}_i$ the set of all strategy profiles of $\sf CG$, that is, the set of outcomes which can be realized when each player $i\in\N$ chooses a strategy in ${\sf\Sigma}_i$. A {\em strategy profile} ${\bm\sigma}=(\sigma_1,\ldots,\sigma_n)$ is then a vector of strategies, where, for each $i\in\N$, $\sigma_i\in{\sf\Sigma}_i$ denotes the choice of player $i$ in $\bm\sigma$. For a strategy profile $\bm\sigma$ and a resource $e\in E$, the value $n_e({\bm\sigma})=\sum_{i\in\N:e\in\sigma_i}w_i$ denotes the {\em congestion} of resource $e$ in $\bm\sigma$, that is, the sum of the weights of all the players choosing $e$ in $\bm\sigma$. The {\em individual cost} of player $i$ in $\bm\sigma$ is defined as $c_i({\bm\sigma})=w_i\sum_{e\in\sigma_i}\ell_e(n_e({\bm\sigma}))$.

Given a finite space of functions ${\cal F}\subseteq\{f:\RP\rightarrow\RP\}$, let ${\cal B}({\cal F})=\{f_k:\RP\rightarrow\RP\ |\ k\in [r]\}$ be a basis for $\cal F$ of cardinality $r$, whose elements (functions) are numbered from $1$ to $r$. We say that $\sf CG$ is defined over $\cal F$ if, for each $e\in E$, it holds that $\ell_e=\sum_{k\in [r]}v_k^e f_k$, where $v_k^e\in\R$ is a scalar. Throughout the paper, we will impose only {\bf\em minimal assumptions} on $\cal F$; in particular, we will assume that any $f\in\cal F$ is non-negative with $f(x)=0$ if and only if $x=0$.

For any $n$-dimensional vector of (positive) weights ${\bm w}=(w_1,\ldots,w_n)$, we denote with ${\cal C}_{\bm w}({\cal F})$ the class of all the weighted congestion games with players' weights induced by $\bm w$ and defined over $\cal F$. Moreover, for a fixed quadruple ${\sf T}_{\bm w}=(\N,{\bm w},E,({\sf\Sigma}_i)_{i\in\N})$, called a {\em congestion model}, the set ${\cal C}_{{\sf T}_{\bm w}}({\cal F})=\{{\sf CG}\in{\cal C}_{\bm w}({\cal F})\ |\ {\sf CG}=({\sf T}_{\bm w},(\ell_e)_{e\in E})\}$ is the set of all the weighted congestion games induced by ${\sf T}_{\bm w}$ and defined over $\cal F$. Note that, since for each game ${\sf CG}\in{\cal C}_{{\sf T}_{\bm w}}({\cal F})$ and $e\in E$ there exist $r$ numbers $v_1^e,\ldots,v_r^e$ such that $\ell_e=\sum_{k\in [r]}v_k^e f_k$, it follows that $\sf CG$ can be specified by the pair $({\sf T}_{\bm w},(v_k^e)_{e\in E,k\in [r]})$. Moreover, it holds that ${\cal C}_{\bm w}({\cal F})=\bigcup_{{\sf T}_{\bm w}}{\cal C}_{{\sf T}_{\bm w}}({\cal F})$. Finally, we denote with ${\sf\Sigma}({\sf T}_{\bm w})$ the set of strategy profiles induced by the congestion model ${\sf T}_{\bm w}$.

A {\em generalized weighted congestion game} is a pair $({\sf CG},\alpha)$ where ${\sf CG}=(\N,(w_i)_{i\in\N},E,({\sf\Sigma}_i)_{i\in\N},(\ell_e)_{e\in E})$ is a weighted congestion game and $\alpha\in\R^{n\times n}$ is an $n$-dimensional square matrix. Game $({\sf CG},\alpha)$ has the same set of players and strategies of $\sf CG$, but the {\em perceived cost} of player $i$ in the strategy profile $\bm\sigma$ is defined as $$\widehat{c}_i({\bm\sigma})=\sum_{j\in\N}\alpha_{ij}c_j({\bm\sigma})=\sum_{j\in\N}\alpha_{ij}w_j\sum_{e\in E:e\in\sigma_j}\ell_e(n_e({\bm\sigma}))=\sum_{e\in E}\sum_{k\in [r]}v^e_k f_k(n_e({\bm\sigma}))\sum_{j\in [n]:e\in\sigma_j}\alpha_{ij}w_j,$$ where $c_i({\bm\sigma})$ is the individual cost that player $i$ experiences in $\bm\sigma$ in the underlying weighted congestion game $\sf CG$.
Note that, when $\alpha$ is the identity matrix, $({\sf CG},\alpha)$ coincides with $\sf CG$, while, in all the other cases, $({\sf CG},\alpha)$ may not be isomorphic to any weighted congestion game, so that the set of generalized weighted congestion games expands that of weighted congestion games.

Given a strategy profile $\bm\sigma$, a player $i\in\N$ and a strategy $x\in{\sf\Sigma}_i$, we denote with $({\bm\sigma}_{-i},x)$ the strategy profile obtained from $\bm\sigma$ when player $i$ changes her strategy from $\sigma_i$ to $x$, while the strategies of all the other players are kept fixed. In particular, for any $\epsilon\geq 0$, the perceived cost suffered by player $i$ in $\bm\sigma$ minus $1+\epsilon$ times the perceived cost suffered by player $i$ in $({\bm\sigma}_{-i},x)$ in a generalized weighted congestion game can be expressed as follows:
\begin{eqnarray*}
\widehat{c}_i({\bm\sigma})-(1+\epsilon)\cdot\widehat{c}_i({\bm\sigma}_{-i},x) & = & \displaystyle\sum_{j\in\N}\alpha_{ij}c_j({\bm\sigma})-(1+\epsilon)\sum_{j\in\N}\alpha_{ij}c_j({\bm\sigma}_{-i},x)\\
& = & \alpha_{ii}w_i\left(\displaystyle\sum_{e\in\sigma_i}\ell_e(n_e({\bm\sigma}))-(1+\epsilon)\sum_{e\in x}\ell_e(n_e({\bm\sigma}_{-i},x))\right)\\
& & +\displaystyle\sum_{j\in\N:j\neq i}\alpha_{ij}w_j\left(\displaystyle\sum_{e\in\sigma_i}\ell_e(n_e({\bm\sigma}))-(1+\epsilon)\sum_{e\in x}\ell_e(n_e({\bm\sigma}_{-i},x))\right)\\
& = & \alpha_{ii}w_i\left(\displaystyle\sum_{e\in\sigma_i\setminus x}\ell_e(n_e({\bm\sigma}))-(1+\epsilon)\sum_{e\in x\setminus\sigma_i}\ell_e(n_e({\bm\sigma})+w_i)\right)\\
& & +\displaystyle\sum_{j\in\N:j\neq i}\alpha_{ij}w_j\left(\displaystyle\sum_{e\in\sigma_i\setminus x}\ell_e(n_e({\bm\sigma}))-(1+\epsilon)\sum_{e\in x\setminus\sigma_i}\ell_e(n_e({\bm\sigma})+w_i)\right).
\end{eqnarray*}

Hence, we get
\begin{multline}\label{deviation}
\widehat{c}_i({\bm\sigma})-(1+\epsilon)\cdot\widehat{c}_i({\bm\sigma}_{-i},x)=\\
\displaystyle\sum_{e\in\sigma_i\setminus x}\ell_e(n_e({\bm\sigma}))\sum_{j\in\N:e\in\sigma_j}\alpha_{ij}w_j-(1+\epsilon)\displaystyle\sum_{e\in x\setminus\sigma_i}
\ell_e(n_e({\bm\sigma})+w_i)\left(\alpha_{ii}w_i+\sum_{j\in\N:e\in\sigma_j}\alpha_{ij}w_j\right).
\end{multline}

Next two definitions formalize the two concepts of approximate equilibria that we will consider throughout the paper.

\begin{definition}
For any $\epsilon\geq 0$, an {\bf $\epsilon$-approximate coarse correlated equilibrium} is a probability distribution $\bm p$ defined over $\sf\Sigma$ such that, for any player $i\in\N$ and strategy $x\in{\sf\Sigma}_i$, it holds that $$\sum_{{\bm \sigma}\in{\sf\Sigma}}p_{\bm \sigma}\cdot \widehat{c}_i({\bm \sigma})\leq (1+\epsilon)\sum_{{\bm \sigma}\in{\sf\Sigma}}p_{\bm \sigma}\cdot \widehat{c}_i({\bm \sigma}_{-i},x),$$
where, for each ${\bm\sigma}\in{\sf\Sigma}$, $p_{\bm\sigma}$ is the probability assigned to $\bm\sigma$ by $\bm p$.
\end{definition}

\begin{definition}
For any $\epsilon\geq 0$, an {\bf $\epsilon$-approximate pure Nash equilibrium} is a strategy profile $\bm\sigma$ such that, for any player $i\in\N$ and strategy $x\in{\sf\Sigma}_i$, it holds that $\widehat{c}_i({\bm\sigma})\leq (1+\epsilon)\cdot\widehat{c}_i({\bm\sigma}_{-i},x).$
\end{definition}

Denote as ${\sf PNE}_\epsilon({\sf CG},\alpha)$ and ${\sf CCE}_\epsilon({\sf CG},\alpha)$, respectively, the set of $\epsilon$-approximate pure Nash equilibria and $\epsilon$-approximate coarse correlated equilibria of the generalized weighted congestion game $({\sf CG},\alpha)$.
It is easy to see that, for any $\epsilon\geq 0$, an $\epsilon$-approximate pure Nash equilibrium $\bm\sigma$ is an $\epsilon$-approximate coarse correlated equilibrium $\bm p$ such that $p_{\bm\sigma}=1$ and $p_{\bm\tau}=0$ for any ${\bm\tau}\in{\sf\Sigma}\setminus\{{\bm\sigma}\}$. So, ${\sf PNE}_\epsilon({\sf CG},\alpha)\subseteq{\sf CCE}_\epsilon({\sf CG},\alpha)$. Moreover, the sets ${\sf PNE}_0({\sf CG},\alpha)$ and ${\sf CCE}_0({\sf CG},\alpha)$ coincide with the sets of pure Nash equilibria and coarse correlated equilibria of $({\sf CG},\alpha)$, respectively.

For an $n$-dimensional {\em non-null} square matrix $\beta\in\R_{\geq 0}^{n\times n}$ and a player $i\in\N$, let $\beta\mbox{-}cost_i:{\sf\Sigma}\rightarrow\RPP$ be the contribution of player $i$ to the definition of the social function which is defined as follows:
\begin{equation*}
\beta\mbox{-}cost_i({\bm\sigma}) =  \sum_{j\in\N}\beta_{ij}c_j({\bm\sigma})=\sum_{e\in E}\sum_{k\in [r]}v^e_k f_k(n_e({\bm\sigma}))\sum_{j\in\N:e\in\sigma_j}\beta_{ij}w_j.
\end{equation*}

Let $\Delta({\sf\Sigma})$ be the set of all the probability distributions defined over $\sf\Sigma$. For a $\bm p\in\Delta({\sf\Sigma})$, the {\em $\beta$-utilitarian social function} is a function $\S:\Delta({\sf\Sigma})\rightarrow\RPP$ such that
\begin{eqnarray*}
\S({\bm p}) & = & \sum_{i\in\N}\mathbb{E}_{{\bm\sigma}\sim{\bm p}}\left[\beta\mbox{-}cost_i({\bm\sigma})\right]\\
& = & \mathbb{E}_{{\bm\sigma}\sim{\bm p}}\left[\sum_{i\in\N}\beta\mbox{-}cost_i({\bm\sigma})\right]\\
& = & \sum_{{\bm\sigma}\in{\sf\Sigma}}p_{{\bm\sigma}}\left(\sum_{e\in E}\sum_{k\in [r]}v^e_k f_k(n_e({\bm\sigma}))\sum_{i\in\N}\sum_{j\in\N:e\in\sigma_j}\beta_{ij}w_j\right)
\end{eqnarray*}
and the {\em $\beta$-egalitarian social function} is a function $\M:\Delta({\sf\Sigma})\rightarrow\RPP$ such that
\begin{eqnarray*}
\M({\bm p}) & = & \max_{i\in\N}\left\{\mathbb{E}_{{\bm\sigma}\sim{\bm p}}\left[\beta\mbox{-}cost_i({\bm\sigma})\right]\right\}\\
& = & \max_{i\in\N}\left\{\sum_{{\bm\sigma}\in{\sf\Sigma}}p_{\bm\sigma}\sum_{e\in E}\sum_{k\in [r]}v^e_k f_k(n_e({\bm\sigma}))\sum_{j\in\N:e\in\sigma_j}\beta_{ij}w_j\right\}.\footnote{There is also another possible (and indeed more traditional) definition for the $\beta$-egalitarian social function, obtained by setting $\M({\bm p}) = \mathbb{E}_{{\bm\sigma}\sim{\bm p}}\left[\max_{i\in\N}\{\beta\mbox{-}cost_i({\bm\sigma})\}\right]$. In such a case, however, the application of the primal-dual method seems to be not so natural, so that the study of this social function remains an interesting open problem at the moment.}
\end{eqnarray*}

Consider the case in which ${\bm p}\in\Delta({\sf\Sigma})$ is indeed a strategy profile ${\bm\sigma}\in{\sf\Sigma}$. When $\beta$ is the identity matrix, $\S$ (resp. $\M$) coincides with the sum (resp. the maximum) of the players' individual costs in the underlying weighted congestion game $\sf CG$, while, when $\beta=\alpha$, $\S$ (resp. $\M$) coincides with the sum (resp. the maximum) of the players' perceived costs in $({\sf CG},\alpha)$. In general, an infinite variety of social functions can be defined by tuning the choice of matrix $\beta$\footnote{One could even relax the constraint $\beta\in\R_{\geq 0}^{n\times n}$ and allow for negative entries in matrix $\beta$ as long as $\sum_{i\in\N}\beta_{ij}\geq 0$ for each $j\in\N$ and $\sum_{i\in\N}\beta_{ij}>0$ for some $j\in\N$ which still guarantees either $\S({\bm\sigma})>0$ and $\M({\bm\sigma})>0$ for each ${\bm\sigma}\in\sf\Sigma$.}. For a function ${\sf SF}\in\{{\sf SUM},{\sf MAX}\}$, we denote with $\bm o$ the {\em social optimum}, that is, any strategy profile minimizing $\beta\mbox{-}{\sf SF}$. Note that, by the properties of the latency functions and the definition of $\beta$\footnote{From now on, we will always assume that $\beta$ is a non-null matrix.}, it follows that $\beta\mbox{-}{\sf SF}({\bm o})>0$.
The $\epsilon$-approximate coarse correlated price of anarchy of $({\sf CG},\alpha)$ under the social function $\beta\mbox{-}\sf SF$ is defined as $${\sf CCPoA}_\epsilon(\beta\mbox{-}{\sf SF},{\sf CG},\alpha)=\max_{{\bm p}\in{\sf CCE}_\epsilon({\sf CG},\alpha)}\frac{\beta\mbox{-}{\sf SF}({\bm p})}{\beta\mbox{-}{\sf SF}({\bm o})},$$
while the $\epsilon$-approximate pure price of anarchy of $({\sf CG},\alpha)$ under the social function $\beta\mbox{-}{\sf SF}$ is defined as
$${\sf PPoA}_\epsilon(\beta\mbox{-}{\sf SF},{\sf CG},\alpha)=\max_{{\bm\sigma}\in{\sf PNE}_\epsilon({\sf CG},\alpha)}\frac{\beta\mbox{-}{\sf SF}({\bm\sigma})}{\beta\mbox{-}{\sf SF}({\bm o})}.$$

For an $n$-dimensional vector of weights ${\bm w}=(w_1,\ldots,w_n)$ and a matrix $\alpha\in\R^{n\times n}$,
we denote with ${\cal C}_{\bm w}({\cal F},\alpha)=\{({\sf CG},\alpha):{\sf CG}\in{\cal C}_{\bm w}({\cal F})\}$ the set of all the generalized weighted congestion games induced by $\bm w$ and $\alpha$ and defined over $\cal F$. Similarly, for any congestion model ${\sf T}_{\bm w}$, one defines the class ${\cal C}_{{\sf T}_{\bm w}}({\cal F},\alpha)$, so as to obtain ${\cal C}_{\bm w}({\cal F},\alpha)=\bigcup_{{\sf T}_{\bm w}}{\cal C}_{{\sf T}_{\bm w}}({\cal F},\alpha)$. The worst-case $\epsilon$-approximate coarse correlated price of anarchy of the class ${\cal C}_{\bm w}({\cal F},\alpha)$ under the social function $\beta\mbox{-}{\sf SF}$ is defined as $${\sf CCPoA}_\epsilon(\beta\mbox{-}{\sf SF},{\cal C}_{\bm w}({\cal F},\alpha))=\sup_{({\sf CG},\alpha)\in{\cal C}_{\bm w}({\cal F},\alpha)}{\sf CCPoA}_\epsilon(\beta\mbox{-}{\sf SF},{\sf CG},\alpha).$$ Similarly, one defines the worst-case $\epsilon$-approximate pure price of anarchy of the class ${\cal C}_{\bm w}({\cal F},\alpha)$ under the social function $\beta\mbox{-}{\sf SF}$.

By ${\sf PNE}_\epsilon({\sf CG},\alpha)\subseteq{\sf CCE}_\epsilon({\sf CG},\alpha)$, it follows that ${\sf PPoA}_\epsilon(\beta\mbox{-}{\sf SF},{\cal C}_{\bm w}({\cal F},\alpha))\leq{\sf CCPoA}_\epsilon(\beta\mbox{-}{\sf SF},{\cal C}_{\bm w}({\cal F},\alpha))$ for any real value $\epsilon\geq 0$, $n$-dimensional vector of weights $\bm w$, finite space of function $\cal F$, pair of matrices $\alpha\in\R^{n\times n}$ and $\beta\in\R_{\geq 0}^{n\times n}$ and function ${\sf SF}\in\{{\sf SUM},{\sf MAX}\}$. Throughout the paper, we will also refer to the worst-case $\epsilon$-approximate pure price of anarchy and to the worst-case $\epsilon$-approximate coarse correlated price of anarchy of subsets of ${\cal C}_{\bm w}({\cal F},\alpha)$ which are naturally defined by restriction.

We conclude this section with an easy, although crucial result, stating that, independently of which is the adopted social function, both the worst-case $\epsilon$-approximate pure price of anarchy and the worst-case $\epsilon$-approximate coarse correlated price of anarchy of a class of generalized weighted congestion games remain the same even if one restricts to only those games in the given class whose social optimum has social value equal to one\footnote{Indeed, such a result implicitly holds for the worst-case $\epsilon$-approximate price of anarchy of any kind of equilibrium.}. To this aim, for any function ${\sf SF}\in\{{\sf SUM},{\sf MAX}\}$ and matrix $\beta\in\R_{\geq 0}^{n\times n}$, let $\overline{{\cal C}}_{\bm w}({\cal F},\alpha)\subset{\cal C}_{\bm w}({\cal F},\alpha)$ be the subset of all the generalized weighted congestion games induced by $\bm w$ and $\alpha$ and defined over $\cal F$ such that the social optimum $\bm o$ satisfies $\beta\mbox{-}{\sf SF}({\bm o})=1$. Similarly, for any congestion model ${\sf T}_{\bm w}$, one defines the class $\overline{{\cal C}}_{{\sf T}_{\bm w}}({\cal F},\alpha)$, so as to obtain $\overline{{\cal C}}_{\bm w}({\cal F},\alpha)=\bigcup_{{\sf T}_{\bm w}}\overline{{\cal C}}_{{\sf T}_{\bm w}}({\cal F},\alpha)$.

\begin{lemma}\label{normalization}
For any real value $\epsilon\geq 0$, $n$-dimensional vector of weights $\bm w$, finite space of functions $\cal F$, pair of matrices $\alpha\in\R^{n\times n}$ and $\beta\in\R_{\geq 0}^{n\times n}$ and function ${\sf SF}\in\{{\sf SUM},{\sf MAX}\}$, it holds that ${\sf PPoA}_\epsilon(\beta\mbox{-}{\sf SF},{\cal C}_{\bm w}({\cal F},\alpha))={\sf PPoA}_\epsilon(\beta\mbox{-}{\sf SF},\overline{{\cal C}}_{\bm w}({\cal F},\alpha))$ and ${\sf CCPoA}_\epsilon(\beta\mbox{-}{\sf SF},{\cal C}_{\bm w}({\cal F},\alpha))={\sf CCPoA}_\epsilon(\beta\mbox{-}{\sf SF},\overline{{\cal C}}_{\bm w}({\cal F},\alpha))$.
\end{lemma}
\begin{proof}
Fix a congestion model ${\sf T}_{\bm w}$, a pair of matrices $\alpha\in\R^{n\times n}$ and $\beta\in\R_{\geq 0}^{n\times n}$ and a function ${\sf SF}\in\{{\sf SUM},{\sf MAX}\}$. The claim directly follows from the fact that, for any game ${\cal G}:=\left({\sf T}_{\bm w},(v_k^e)_{e\in E,k\in [r]},\alpha\right)\in{\cal C}_{{\sf T}_{\bm w}}({\cal F},\alpha)$ such that $\beta\mbox{-}{\sf SF}({\bm o}):=x>0$, there always exists a game $\overline{{\cal G}}:=\left({\sf T}_{\bm w},(\overline{v}_k^e)_{e\in E,k\in [r]},\alpha\right)\in\overline{{\cal C}}_{{\sf T}_{\bm w}}({\cal F},\alpha)$, obtained by setting $\overline{v}^e_k=v^e_k/x$, such that, for any ${\bm\sigma}\in{\sf\Sigma}({\sf T}_{\bm w})$, it holds that
$$\sum_{e\in E}\sum_{k\in [r]}v^e_k f_k(n_e({\bm\sigma}))\sum_{i\in\N}\sum_{j\in\N:e\in\sigma_j}\beta_{ij}w_j=x \sum_{e\in E}\sum_{k\in [r]}\overline{v}^e_k f_k(n_e({\bm\sigma}))\sum_{i\in\N}\sum_{j\in\N:e\in\sigma_j}\beta_{ij}w_j$$

and that

$$\max_{i\in\N}\sum_{j\in\N}\beta_{ij}w_j\sum_{e\in\sigma_j}\sum_{k\in [r]}v^e_k f_k(n_e({\bm\sigma}))=x\cdot\max_{i\in\N}\sum_{j\in\N}\beta_{ij}w_j\sum_{e\in\sigma_j}\sum_{k\in [r]}\overline{v}^e_k f_k(n_e({\bm\sigma})).$$

Moreover, for any ${\bm\sigma}\in{\sf\Sigma}({\sf T}_{\bm w})$ and $i\in\N$, it holds that

$$\sum_{e\in E}\sum_{k\in [r]}v^e_k f_k(n_e({\bm\sigma}))\sum_{j\in\N:e\in\sigma_j}\alpha_{ij}w_j=x \sum_{e\in E}\sum_{k\in [r]}\overline{v}^e_k f_k(n_e({\bm\sigma}))\sum_{j\in\N:e\in\sigma_j}\alpha_{ij}w_j.$$

That is, for any strategy profile ${\bm\sigma}\in{\sf\Sigma}({\sf T}_{\bm w})$, the social value of $\bm\sigma$ in game $\cal G$ is equal to $x$ times the social value of $\bm\sigma$ in game $\overline{\cal G}$, independently of which is the adopted social function. Moreover, for any strategy profile ${\bm\sigma}\in{\sf\Sigma}({\sf T}_{\bm w})$ and any $i\in\N$, the perceived cost of player $i$ in $\bm\sigma$ in game $\cal G$ is equal to $x$ times the perceived cost of player $i$ in $\bm\sigma$ in game $\overline{\cal G}$. This implies that $\cal G$ and $\overline{\cal G}$ have the same set of equilibria (whatever the concept of equilibrium is defined) and that the ratio between any linear combination of the social values of any set of strategy profiles is the same in both games.\qed
\end{proof}

\section{The Main Result}\label{sec-main}

Our main result is the proof of the following general theorem.

\begin{theorem}\label{mainth}
For any real value $\epsilon\geq 0$, $n$-dimensional vector of weights $\bm w$, finite space of functions $\cal F$, pair of matrices $\alpha\in\R^{n\times n}$ and $\beta\in\R_{\geq 0}^{n\times n}$ and function ${\sf SF}\in\{{\sf SUM},{\sf MAX}\}$, it holds that ${\sf PPoA}_\epsilon(\beta\mbox{-}{\sf SF},{\cal C}_{\bm w}({\cal F},\alpha))={\sf CCPoA}_\epsilon(\beta\mbox{-}{\sf SF},{\cal C}_{\bm w}({\cal F},\alpha))$. Moreover, the value ${\sf PPoA}_\epsilon(\beta\mbox{-}{\sf SF},{\cal C}_{\bm w}({\cal F},\alpha))$ can always be determined via the primal-dual method.
\end{theorem}

\begin{proof}
Fix a real value $\epsilon\geq 0$, an $n$-dimensional vector of weights $\bm w$, a finite space of functions $\cal F$, a pair of matrices $\alpha\in\R^{n\times n}$ and $\beta\in\R_{\geq 0}^{n\times n}$ and a function ${\sf SF}\in\{{\sf SUM},{\sf MAX}\}$. We prove the claim in four steps.

\

{\noindent\bf Step 1)} Definition of the {\em representative} congestion model ${\sf T}^*_{\bm w}$.

\

Let ${\sf T}_{\bm w}^*=(\N,{\bm w},E^*,({\sf\Sigma}_i^*)_{i\in\N})$ be a congestion model such that
\begin{enumerate}
\item ${\sf\Sigma}^*_i=\{\sigma_i^*,o_i^*\}$ for each $i\in\N$, i.e., each player $i\in\N$ has exactly two strategies denoted as $\sigma^*_i$ and $o^*_i$;
\item the set of resources $E^*$ and the strategies $\sigma^*_i$ and $o^*_i$ for each $i\in\N$ are properly defined in such a way that, for each $P,Q\subseteq\N$, there exists exactly one resource $e(P,Q)\in E^*$ for which it holds that $\{i\in\N\ |\ e(P,Q)\in\sigma^*_i\}=P$ and $\{i\in\N\ |\ e(P,Q)\in o^*_i\}=Q$. Hence, $|E^*|=2^n\cdot 2^n=4^n$.
\end{enumerate}
Intuitively, the representative congestion model ${\sf T}_{\bm w}^*$ is defined in such a way that the pair of strategy profiles ${\bm\sigma}^*=(\sigma^*_1,\ldots,\sigma^*_n)$ and ${\bm o}^*=(o^*_1,\ldots,o^*_n)$ is able to encompass all possible configurations of congestions that may arise in any pair of strategy profiles and for any congestion model induced by $\bm w$. In particular, the following fundamental property holds.
\begin{property}\label{prop1}
For any congestion model ${\sf T}_{\bm w}=(\N,{\bm w},E,({\sf\Sigma}_i)_{i\in\N})$, resource $e\in E$ and pair of profiles ${\bm\sigma}',{\bm\sigma}''\in{\sf\Sigma}({\sf T}_{\bm w})$, there always exists a resource $\overline{e}\in E^*$ such that $\{i\in [n]\ |\ e\in\sigma'_i\}=\{i\in [n]\ |\ \overline{e}\in\sigma^*_i\}$ and $\{i\in [n]\ |\ e\in\sigma''_i\}=\{i\in [n]\ |\ \overline{e}\in o^*_i\}$.
\end{property}
\begin{proof}
Fix a congestion model ${\sf T}_{\bm w}=(\N,{\bm w},E,({\sf\Sigma}_i)_{i\in\N})$, a resource $e\in E$ and pair of profiles ${\bm\sigma}',{\bm\sigma}''\in{\sf\Sigma}({\sf T}_{\bm w})$. Let $\{i\in [n]\ |\ e\in\sigma'_i\}:=P$ and $\{i\in [n]\ |\ e\in\sigma''_i\}:=Q$. To prove the claim, it suffices choosing $\overline{e}=e(P,Q)$.\qed
\end{proof}

\

{\noindent\bf Step 2)} Definition of a primal-dual formulation for ${\sf PPoA}_\epsilon(\beta\mbox{-}{\sf SF},\overline{{\cal C}}_{\bm w}({\cal F},\alpha))$.

\

Fix a function ${\sf SF}\in\{{\sf SUM},{\sf MAX}\}$. Our aim is to use the optimal solution of a linear program ${\sf PP_{PNE}}({\sf SF},{\sf T}_{\bm w}^*,{\bm \sigma}^*,{\bm o}^*)$ to achieve an upper bound on the worst-case $\epsilon$-approximate pure price of anarchy of any game in $\overline{{\cal C}}_{{\sf T}_{\bm w}^*}({\cal F},\alpha)$ under the restriction that the latency functions are suitably tuned so as to make ${\bm \sigma}^*$ the worst $\epsilon$-approximate pure Nash equilibrium and ${\bm o}^*$ a social optimum (of social value $1$).
The linear program ${\sf PP_{PNE}}({\sf SUM},{\sf T}_{\bm w}^*,{\bm \sigma}^*,{\bm o}^*)$ for the $\beta$-utilitarian social function is defined as follows.

\begin{displaymath}
\begin{array}{ll}
maximize \displaystyle\sum_{e\in E^*}\sum_{k\in [r]}v_k^e f_k(n_e({\bm\sigma}^*))\sum_{i\in\N}\sum_{j\in\N:e\in\sigma^*_j}\beta_{ij}w_j\\\vspace{0.1cm}
subject\ to\\\vspace{0.1cm}
\displaystyle\sum_{e\in\sigma^*_i\setminus o^*_i}\sum_{k\in [r]}v^e_k f_k(n_e({\bm\sigma}^*))\sum_{j\in\N:e\in\sigma^*_j}\alpha_{ij}w_j\\
\ \ \ \ -(1+\epsilon)\displaystyle\sum_{e\in o^*_i\setminus\sigma^*_i\in E}\sum_{k\in [r]}v^e_k f_k(n_e({\bm\sigma}^*)+w_i)\left(\alpha_{ii}w_i+\sum_{j\in\N:e\in\sigma^*_j}\alpha_{ij}w_j\right)\leq 0, & \ \ \forall i\in\N\\\vspace{0.1cm}
\displaystyle\sum_{e\in E^*}\sum_{k\in [r]}v_k^e f_k(n_e({\bm o}^*))\sum_{i\in\N}\sum_{j\in\N:e\in o^*_j}\beta_{ij}w_j \leq 1,\\\vspace{0.1cm}
v^e_k\geq 0, & \ \ \forall e\in E^*,k\in [r]
\end{array}
\end{displaymath}
The first $n$ constraints guarantee that no player can lower her perceived cost of a factor more than $1+\epsilon$ by switching to the strategy she uses in the social optimum ${\bm o}^*$ (see Equation (\ref{deviation})), while the last constraint normalizes to at most $1$ the value $\S({\bm o}^*)$.

The dual program ${\sf DP_{PNE}}({\sf SUM},{\sf T}_{\bm w}^*,{\bm \sigma}^*,{\bm o}^*)$ is the following (we associate a variable $y_i$ with the $i$th constraint of the first $n$ ones and a variable $\gamma$ with the normalizing constraint).
\begin{displaymath}
\begin{array}{ll}
minimize\ \gamma\\\vspace{0.1cm}
subject\ to\\\vspace{0.1cm}
\displaystyle\sum_{i\in\N:e\in\sigma^*_i\setminus o^*_i}y_i f_k(n_e({\bm \sigma}^*))\sum_{j\in\N:e\in\sigma^*_j}\alpha_{ij}w_j\\
\ \ -(1+\epsilon)\displaystyle\sum_{i\in\N:e\in o^*_i\setminus\sigma^*_i}y_i f_k(n_e({\bm \sigma}^*)+w_i)\left(\alpha_{ii}w_i+\sum_{j\in\N:e\in\sigma^*_j}\alpha_{ij}w_j\right)\\
\ \ +\displaystyle\gamma f_k(n_e({\bm o}^*))\sum_{i\in\N}\sum_{j\in\N:e\in o^*_j}\beta_{ij}w_j\geq f_k(n_e({\bm\sigma}^*))\sum_{i\in\N}\sum_{j\in\N:e\in\sigma^*_j}\beta_{ij}w_j, & \ \ \forall e\in E^*,k\in [r]\\\vspace{0.1cm}
y_i\geq 0, & \ \ \forall i\in\N\\\vspace{0.1cm}
\gamma\geq 0
\end{array}
\end{displaymath}

Similarly, the linear program ${\sf PP_{PNE}}({\sf MAX},{\sf T}_{\bm w}^*,{\bm \sigma}^*,{\bm o}^*)$ for the $\beta$-egalitarian social function is defined as follows.

\begin{displaymath}
\begin{array}{ll}
maximize\ t\\\vspace{0.1cm}
subject\ to\\\vspace{0.1cm}
\displaystyle\sum_{e\in\sigma^*_i\setminus o^*_i}\sum_{k\in [r]}v^e_k f_k(n_e({\bm\sigma}^*))\sum_{j\in\N:e\in\sigma^*_j}\alpha_{ij}w_j\\
\ \ \ \ -(1+\epsilon)\displaystyle\sum_{e\in o^*_i\setminus\sigma^*_i\in E}\sum_{k\in [r]}v^e_k f_k(n_e({\bm\sigma}^*)+w_i)\left(\alpha_{ii}w_i+\sum_{j\in\N:e\in\sigma^*_j}\alpha_{ij}w_j\right)\leq 0, & \ \ \forall i\in\N\\\vspace{0.1cm}
\displaystyle\sum_{e\in E^*}\sum_{k\in [r]}v_k^e f_k(n_e({\bm\sigma}^*))\sum_{j\in\N:e\in\sigma^*_j}\beta_{1j}w_j= t, \\\vspace{0.1cm}
\displaystyle\sum_{e\in E^*}\sum_{k\in [r]}v_k^e f_k(n_e({\bm\sigma}^*))\sum_{j\in\N:e\in\sigma^*_j}\beta_{ij}w_j\leq t, & \ \ \forall i\in\N\setminus\{1\}\\\vspace{0.1cm}
\displaystyle\sum_{e\in E^*}\sum_{k\in [r]}v_k^e f_k(n_e({\bm o}^*))\sum_{j\in\N:e\in o^*_j}\beta_{ij}w_j\leq 1, & \ \ \forall i\in\N\\\vspace{0.1cm}
v^e_k\geq 0, & \ \ \forall e\in E^*,k\in [r]\\
t\geq 0
\end{array}
\end{displaymath}
Here, again the first $n$ constraints guarantee that no player can lower her perceived cost of a factor more than $1+\epsilon$ by switching to the strategy she uses in the social optimum ${\bm o}^*$. The next $n$ constraints impose that the maximum value in the social function $\M({\bm\sigma}^*)$ is attained by player $1$ (this hypothesis is without loss of generality up to a renumbering of the players) and has value $t$ (which is the objective function to be maximized), while the last $n$ constraints normalizes to at most $1$ the value $\M({\bm o}^*)$.

The dual program ${\sf DP_{PNE}}({\sf MAX},{\sf T}_{\bm w}^*,{\bm \sigma}^*,{\bm o}^*)$ is the following (we associate variables $y_i$, $z_i$ and $\gamma_i$ with the $i$th constraint of the first, the middle and the last family of $n$ constraints, respectively).
\begin{displaymath}
\begin{array}{ll}
minimize\displaystyle\sum_{i\in\N}\gamma_i\\\vspace{0.1cm}
subject\ to\\\vspace{0.1cm}
\displaystyle\sum_{i\in\N:e\in\sigma^*_i\setminus o^*_i}y_i f_k(n_e({\bm \sigma}^*))\sum_{j\in\N:e\in\sigma^*_j}\alpha_{ij}w_j\\
\ \ -(1+\epsilon)\displaystyle\sum_{i\in\N:e\in o^*_i\setminus\sigma^*_i}y_i f_k(n_e({\bm \sigma}^*)+w_i)\left(\alpha_{ii}w_i+\sum_{j\in\N:e\in\sigma^*_j}\alpha_{ij}w_j\right)\\
\ \ +\displaystyle f_k(n_e({\bm\sigma}^*))\sum_{i\in\N}z_i\sum_{j\in\N:e\in \sigma^*_j}\beta_{ij}w_j\\
\ \ +\displaystyle f_k(n_e({\bm o}^*))\sum_{i\in\N}\gamma_i\sum_{j\in\N:e\in o^*_j}\beta_{ij}w_j\geq 0, & \ \ \forall e\in E^*,k\in [r]\\\vspace{0.1cm}
\displaystyle\sum_{i\in N}z_i\leq -1\\\vspace{0.1cm}
y_i,z_i,\gamma_i\geq 0, & \ \ \forall i\in\N
\end{array}
\end{displaymath}

We stress that, being all the values $\epsilon$, $(w_i)_{i\in\N}$, $(\alpha_{ij},\beta_{ij})_{i,j\in\N}$, $n_e({\bm\sigma}^*)$ and $n_e({\bm o}^*)$ fixed constants in the proposed formulations, ${\sf PP_{PNE}}({\sf SUM},{\sf T}_{\bm w}^*,{\bm \sigma}^*,{\bm o}^*)$ is a linear program defined over the variables $(v^e_k)_{e\in E^*,k\in [r]}$ and ${\sf PP_{PNE}}({\sf MAX},{\sf T}_{\bm w}^*,{\bm \sigma}^*,{\bm o}^*)$ is a linear program defined over the variables $(v^e_k)_{e\in E^*,k\in [r]}$ and $t$, as needed. Note that, for ${\sf SF}\in\{{\sf SUM},{\sf MAX}\}$, ${\sf PP_{PNE}}({\sf SF},{\sf T}_{\bm w}^*,{\bm \sigma}^*,{\bm o}^*)$ is, in general, under-constrained. In fact, in order to assure that ${\bm\sigma}^*$ and ${\bm o}^*$ are the worst $\epsilon$-approximate pure Nash equilibrium and the social optimum, respectively, one should guarantee $\beta\mbox{-}{\sf SF}({\bm\sigma}^*)\geq\beta\mbox{-}{\sf SF}({\bm\sigma})$ for each other $\epsilon$-approximate pure Nash equilibrium ${\bm\sigma}\in{\sf\Sigma}^*$, if any, and $\beta\mbox{-}{\sf SF}({\bm o}^*)\leq\beta\mbox{-}{\sf SF}({\bm\sigma})$ for each ${\bm\sigma}\in{\sf\Sigma}^*$. Moreover, the normalizing constraints have also been relaxed so as to assure $\beta\mbox{-}{\sf SF}({\bm o}^*)\leq 1$ rather than $\beta\mbox{-}{\sf SF}({\bm o}^*) = 1$. Anyway, as we will discuss in the proof of Lemma \ref{lemma1}, either removing or relaxing these constraints can only worsen the resulting upper bounds.

The significance of the previously defined pairs of primal-dual formulations is witnessed by the following lemma which states that the value of an optimal solution to ${\sf PP_{PNE}}({\sf SF},{\sf T}_{\bm w}^*,{\bm \sigma}^*,{\bm o}^*)$ provides an upper bound on ${\sf PPoA}_\epsilon(\beta\mbox{-}{\sf SF},\overline{{\cal C}}_{\bm w}({\cal F},\alpha))$.

\begin{lemma}\label{lemma1}
For a fixed ${\sf SF}\in\{{\sf SUM},{\sf MAX}\}$, let $\overline{x}$ be the value of an optimal solution to ${\sf PP_{PNE}}({\sf SF},{\sf T}_{\bm w}^*,{\bm \sigma}^*,{\bm o}^*)$ when this linear problem is not unlimited, otherwise let $\overline{x}=\infty$. Then ${\sf PPoA}_\epsilon(\beta\mbox{-}{\sf SF},\overline{{\cal C}}_{\bm w}({\cal F},\alpha))\leq \overline{x}$.
\end{lemma}
\begin{proof}
We first show that ${\sf PP_{PNE}}({\sf SUM},{\sf T}_{\bm w}^*,{\bm \sigma}^*,{\bm o}^*)$ and ${\sf PP_{PNE}}({\sf MAX},{\sf T}_{\bm w}^*,{\bm \sigma}^*,{\bm o}^*)$ are both feasible. In fact, fixed an index $k^*\in [r]$, the following solution
\begin{displaymath}
v_k^e=\left\{
\begin{array}{cl}
\left(n f_k(w_j)w_j\sum_{i\in\N}\beta_{ij}\right)^{-1} & \textrm{if }k=k^*\textrm{ and }e\in\left\{e(\{j\},\emptyset),e(\emptyset,\{j\})\right\}\textrm{ for some }j\in\N,\\
0 & \textrm{otherwise}
\end{array}
\right.
\end{displaymath}
is feasible for ${\sf PP_{PNE}}({\sf SUM},{\sf T}_{\bm w}^*,{\bm \sigma}^*,{\bm o}^*)$ and yields an objective value equal to $1$. Similarly, assuming, for instance, that the players are numbered in such a way that $\sum_{i\in\N}\beta_{1i}\geq\sum_{i\in\N}\beta_{ji}$ for each $j\in\N\setminus\{1\}$, the solution with $t=1$ and
\begin{displaymath}
v_k^e=\left\{
\begin{array}{cl}
\left(f_k(w_j)w_j\sum_{i\in\N}\beta_{1i}\right)^{-1} & \textrm{if }k=k^*\textrm{ and }e\in\left\{e(\{j\},\emptyset),e(\emptyset,\{j\})\right\}\textrm{ for some }j\in\N,\\
0 & \textrm{otherwise}
\end{array}
\right.
\end{displaymath}
is feasible for ${\sf PP_{PNE}}({\sf MAX},{\sf T}_{\bm w}^*,{\bm \sigma}^*,{\bm o}^*)$ and yields an objective value equal to $1$.
Hence, for any ${\sf SF}\in\{{\sf SUM},{\sf MAX}\}$, exactly one of the two cases included in the claim may occur.

If ${\sf PP_{PNE}}({\sf SF},{\sf T}_{\bm w}^*,{\bm \sigma}^*,{\bm o}^*)$ is unlimited, then $\overline{x}=\infty$ and the claim is trivially true.

So, we can assume that ${\sf PP_{PNE}}({\sf SF},{\sf T}_{\bm w}^*,{\bm \sigma}^*,{\bm o}^*)$ admits an optimal solution of value $\overline{x}$. As we have already observed, ${\sf PP_{PNE}}({\sf SF},{\sf T}_{\bm w}^*,{\bm \sigma}^*,{\bm o}^*)$ may be under-constrained. Nevertheless, recall that we are only interested in an upper bound on the worst-case $\epsilon$-approximate pure price of anarchy of the class $\overline{{\cal C}}_{{\sf T}_{\bm w}^*}({\cal F},\alpha)$ attainable when the latency functions are suitably tuned so as to make ${\bm \sigma}^*$ the worst $\epsilon$-approximate pure Nash equilibrium and ${\bm o}^*$ a social optimum (of social value $1$). Let us denote with $\widehat{{\cal C}}_{{\sf T}_{\bm w}^*}({\cal F},\alpha)$ such a subclass of $\overline{{\cal C}}_{{\sf T}_{\bm w}^*}({\cal F},\alpha)$. Hence, since once fixed the profiles ${\bm\sigma}^*$ and ${\bm o}^*$ any game in $\widehat{{\cal C}}_{{\sf T}_{\bm w}^*}({\cal F},\alpha)$ can be specified by a particular choice of the values $v_k^e$, and because the removal or the relaxation of some constraints in a maximization problem can only increase the value of the optimal solution, we obtain that the optimal solution to ${\sf PP_{PNE}}({\sf SF},{\sf T}_{\bm w}^*,{\bm \sigma}^*,{\bm o}^*)$ yields an upper bound on the worst-case $\epsilon$-approximate pure price of anarchy of the class $\widehat{{\cal C}}_{{\sf T}_{\bm w}^*}({\cal F},\alpha)$. That is, ${\sf PPoA}_\epsilon(\beta\mbox{-}{\sf SF},\widehat{{\cal C}}_{{\sf T}^*_{\bm w}}({\cal F},\alpha))\leq \overline{x}$. Moreover, since the optimal solution to ${\sf PP_{PNE}}({\sf SF},{\sf T}_{\bm w}^*,{\bm \sigma}^*,{\bm o}^*)$ has value $\overline{x}$, then, by the Strong Duality Theorem, each optimal solution $({\bm y}^*,\gamma^*)$ to ${\sf DP_{PNE}}({\sf SF},{\sf T}_{\bm w}^*,{\bm \sigma}^*,{\bm o}^*)$ satisfies $\overline{x}=\gamma^*$. By Property \ref{prop1}, the particular combinatorial structure of the pair ${\bm\sigma}^*$ and ${\bm o}^*$ implies that, for any alternative pair of strategy profiles ${\bm\sigma}$ and ${\bm o}$ in ${\sf T}_{\bm w}^*$, the set of constraints of the dual program ${\sf DP_{PNE}}({\sf SF},{\sf T}_{\bm w}^*,{\bm \sigma},{\bm o})$ is a subset of that of ${\sf DP_{PNE}}({\sf SF},{\sf T}_{\bm w}^*,{\bm \sigma}^*,{\bm o}^*)$. This implies that any optimal solution $({\bm y},\gamma)$ to ${\sf DP_{PNE}}({\sf SF},{\sf T}_{\bm w}^*,{\bm \sigma},{\bm o})$ must obey $\gamma\leq\gamma^*$. Thus, one can claim that $\gamma^*=\overline{x}$ is indeed an upper bound on the worst-case $\epsilon$-approximate pure price of anarchy of the class $\overline{{\cal C}}_{{\sf T}_{\bm w}^*}({\cal F},\alpha)$, that is, ${\sf PPoA}_\epsilon(\beta\mbox{-}{\sf SF},\overline{{\cal C}}_{{\sf T}^*_{\bm w}}({\cal F},\alpha))\leq \overline{x}$. Note also that, again by Property \ref{prop1}, for any other congestion model ${\sf T}_{\bm w}=(\N,{\bm w},E,({\sf\Sigma}_i)_{i\in\N})$, the pair of primal-dual formulations ${\sf PP_{PNE}}({\sf SF},{\sf T}_{\bm w},{\bm \sigma},{\bm o})$ and ${\sf DP_{PNE}}({\sf SF},{\sf T}_{\bm w},{\bm \sigma},{\bm o})$ induced by any pair of strategy profiles ${\bm\sigma},{\bm o}\in{\sf \Sigma}({\sf T}_{\bm w})$ are such that the set of constraints of ${\sf DP_{PNE}}({\sf SF},{\sf T}_{\bm w},{\bm \sigma},{\bm o})$ is again a subset of that of ${\sf DP_{PNE}}({\sf SF},{\sf T}_{\bm w}^*,{\bm \sigma}^*,{\bm o}^*)$ and this implies that $\gamma^*=\overline{x}$ is even an upper bound on the worst-case $\epsilon$-approximate pure price of anarchy of the whole class $\overline{{\cal C}}_{\bm w}({\cal F},\alpha)=\bigcup_{{\sf T}_{\bm w}}\overline{{\cal C}}_{{\sf T}_{\bm w}}({\cal F},\alpha)$, that is, ${\sf PPoA}_\epsilon(\beta\mbox{-}{\sf SF},\overline{{\cal C}}_{\bm w}({\cal F},\alpha))\leq \overline{x}$.\qed
\end{proof}

\

{\noindent\bf Step 3)} Proof of existence of a game $({\sf CG},\alpha)\in\overline{{\cal C}}_{\bm w}({\cal F},\alpha)$ such that ${\sf PPoA}_\epsilon(\beta\mbox{-}{\sf SF},{\sf CG},\alpha)=\overline{x}$.

\begin{lemma}\label{lemma2}
For a fixed ${\sf SF}\in\{{\sf SUM},{\sf MAX}\}$, let $\overline{x}$ be the value of an optimal solution to ${\sf PP_{PNE}}({\sf SF},{\sf T}_{\bm w}^*,{\bm \sigma}^*,{\bm o}^*)$ when this linear problem is not unlimited, otherwise let $\overline{x}=\infty$. Then ${\sf PPoA}_\epsilon(\beta\mbox{-}{\sf SF},\overline{{\cal C}}_{\bm w}({\cal F},\alpha))=\overline{x}$.
\end{lemma}
\begin{proof}
Assume that ${\sf PP_{PNE}}({\sf SF},{\sf T}_{\bm w}^*,{\bm \sigma}^*,{\bm o}^*)$ admits a feasible solution $\widehat{{\sf SOL}}_{\sf SUM}=(\widehat{v}_k^e)_{e\in E^*,k\in [r]}$ or $\widehat{{\sf SOL}}_{\sf MAX}=\left((\widehat{v}_k^e)_{e\in E^*,k\in [r]},\widehat{t}\right)$, both of value $\widehat{x}$, depending on which is the value of $\sf SF$. Consider the game $({\sf CG},\alpha)$, where ${\sf CG}=({\sf T}^*_{\bm w},(\widehat{v}_k^e)_{e\in E^*,k\in [r]})$ is defined by the representative congestion model ${\sf T}^*_{\bm w}$ coupled with the values $(\widehat{v}_k^e)_{e\in E^*,k\in [r]}$. Since, for any ${\sf SF}\in\{{\sf SUM},{\sf MAX}\}$, $\widehat{{\sf SOL}}_{\sf SF}$ is feasible for ${\sf PP_{PNE}}({\sf SF},{\sf T}_{\bm w}^*,{\bm \sigma}^*,{\bm o}^*)$, it follows that ${\bm\sigma}^*$ is an $\epsilon$-approximate pure Nash equilibrium for $({\sf CG},\alpha)$ such that $\beta\mbox{-}{\sf SF}({\bm\sigma}^*)=\widehat{x}$. This implies ${\sf PPoA}_\epsilon(\beta\mbox{-}{\sf SF},\overline{{\cal C}}_{\bm w}({\cal F},\alpha))\geq\widehat{x}$ (recall, in fact, that $\beta\mbox{-}{\sf SF}({\bm o}^*)\leq 1$).

In the case in which ${\sf PP_{PNE}}({\sf SF},{\sf T}_{\bm w}^*,{\bm \sigma}^*,{\bm o}^*)$ admits an optimal solution $\overline{\sf SOL}_{\sf SF}$ of value $\overline{x}$, by the above argument, it follows that ${\sf PPoA}_\epsilon(\beta\mbox{-}{\sf SF},\overline{{\cal C}}_{\bm w}({\cal F},\alpha))\geq\overline{x}$, which, together with Lemma \ref{lemma1}, implies the claim. In the case in which ${\sf PP_{PNE}}({\sf SF},{\sf T}_{\bm w}^*,{\bm \sigma}^*,{\bm o}^*)$ is unlimited, then, for any $x\in\R$, there exists a feasible solution ${\sf SOL}_{\sf SF}$ to ${\sf PP_{PNE}}({\sf SF},{\sf T}_{\bm w}^*,{\bm \sigma}^*,{\bm o}^*)$ of value at least $x$, which implies that, for any $x\in\R$, it holds that ${\sf PPoA}_\epsilon(\beta\mbox{-}{\sf SF},\overline{{\cal C}}_{\bm w}({\cal F},\alpha))\geq x$.\qed
\end{proof}

\

{\noindent\bf Step 4)} Definition of a primal-dual formulation for ${\sf CCPoA}_\epsilon(\beta\mbox{-}{\sf SF},\overline{{\cal C}}_{\bm w}({\cal F},\alpha))$ and proof of the ``Extension Lemma''.

\

Fix a congestion model ${\sf T}_{\bm w}=(\N,{\bm w},E,({\sf\Sigma}_i)_{i\in\N})$, a probability distribution ${\bm p}\in\Delta({\sf\Sigma}({\sf T}_{\bm w}))$ and a strategy profile ${\bm o}\in{\sf\Sigma}({\sf T}_{\bm w})$. We define the following primal program ${\sf PP_{CCE}}({\sf SUM},{\sf T}_{\bm w},{\bm p},{\bm o})$ for the $\beta$-utilitarian social function.

\begin{displaymath}
\begin{array}{ll}
maximize \displaystyle\sum_{{\bm\sigma}\in{\sf\Sigma}}p_{\bm\sigma}\sum_{e\in E}\sum_{k\in [r]}v_k^e f_k(n_e({\bm\sigma}))\sum_{i\in\N}\sum_{j\in\N:e\in\sigma_j}\beta_{ij}w_j\\\vspace{0.1cm}
subject\ to\\\vspace{0.1cm}
\displaystyle\sum_{{\bm\sigma}\in{\sf\Sigma}}p_{\bm\sigma}\sum_{e\in\sigma_i\setminus o_i}\sum_{k\in [r]}v^e_k f_k(n_e({\bm\sigma}))\sum_{j\in\N:e\in\sigma_j}\alpha_{ij}w_j\\
\ \ \ \ -(1+\epsilon)\displaystyle\sum_{{\bm\sigma}\in{\sf\Sigma}}p_{\bm\sigma}\sum_{e\in o_i\setminus\sigma_i}\sum_{k\in [r]}v^e_k f_k(n_e({\bm\sigma})+w_i)\left(\alpha_{ii}w_i+\sum_{j\in\N:e\in\sigma_j}\alpha_{ij}w_j\right)\leq 0, & \ \ \forall i\in\N\\\vspace{0.1cm}
\displaystyle\sum_{e\in E}\sum_{k\in [r]}v_k^e f_k(n_e({\bm o}))\sum_{i\in\N}\sum_{j\in\N:e\in o_j}\beta_{ij}w_j \leq 1,\\\vspace{0.1cm}
v^e_k\geq 0, & \ \ \forall e\in E,k\in [r]
\end{array}
\end{displaymath}

The dual program ${\sf DP_{CCE}}({\sf SUM},{\sf T}_{\bm w},{\bm p},{\bm o})$ is the following (again, we associate a variable $y_i$ with the $i$th constraint of the first $n$ ones and a variable $\gamma$ with the normalizing constraint).
\begin{displaymath}
\begin{array}{ll}
minimize\ \gamma\\\vspace{0.1cm}
subject\ to\\\vspace{0.1cm}
\displaystyle\sum_{{\bm\sigma}\in{\sf\Sigma}}p_{\bm\sigma}\sum_{i\in\N:e\in\sigma_i\setminus o_i}y_i f_k(n_e({\bm \sigma}))\sum_{j\in\N:e\in\sigma_j}\alpha_{ij}w_j\\
\ \ \ \ -(1+\epsilon)\displaystyle\sum_{{\bm\sigma}\in{\sf\Sigma}}p_{\bm\sigma}\sum_{i\in\N:e\in o_i\setminus\sigma_i}y_i f_k(n_e({\bm \sigma})+w_i)\left(\alpha_{ii}w_i+\sum_{j\in\N:e\in\sigma_j}\alpha_{ij}w_j\right)\\
\ \ +\displaystyle\gamma f_k(n_e({\bm o}))\sum_{i\in\N}\sum_{j\in\N:e\in o_j}\beta_{ij}w_j\geq \sum_{{\bm\sigma}\in{\sf\Sigma}}p_{\bm\sigma}f_k(n_e({\bm\sigma}))\sum_{i\in\N}\sum_{j\in\N:e\in\sigma_j}\beta_{ij}w_j, & \ \ \forall e\in E,k\in [r]\\\vspace{0.1cm}
y_i\geq 0, & \ \ \forall i\in\N\\
\gamma\geq 0
\end{array}
\end{displaymath}

Similarly, for the $\beta$-egalitarian social function, the primal program ${\sf PP_{CCE}}({\sf MAX},{\sf T}_{\bm w},{\bm p},{\bm o})$ is defined as follows.

\begin{displaymath}
\begin{array}{ll}
maximize\ t\\\vspace{0.1cm}
subject\ to\\\vspace{0.1cm}
\displaystyle\sum_{{\bm\sigma}\in{\sf\Sigma}}p_{\bm\sigma}\sum_{e\in\sigma_i\setminus o_i}\sum_{k\in [r]}v^e_k f_k(n_e({\bm\sigma}))\sum_{j\in\N:e\in\sigma_j}\alpha_{ij}w_j\\
\ \ \ \ -(1+\epsilon)\displaystyle\sum_{{\bm\sigma}\in{\sf\Sigma}}p_{\bm\sigma}\sum_{e\in o_i\setminus\sigma_i}\sum_{k\in [r]}v^e_k f_k(n_e({\bm\sigma})+w_i)\left(\alpha_{ii}w_i+\sum_{j\in\N:e\in\sigma_j}\alpha_{ij}w_j\right)\leq 0, & \ \ \forall i\in\N\\\vspace{0.1cm}
\displaystyle\sum_{{\bm\sigma}\in{\sf\Sigma}}p_{\bm\sigma}\sum_{e\in E}\sum_{k\in [r]}v_k^e f_k(n_e({\bm\sigma}))\sum_{j\in\N:e\in\sigma_j}\beta_{1j}w_j= t, \\\vspace{0.1cm}
\displaystyle\sum_{{\bm\sigma}\in{\sf\Sigma}}p_{\bm\sigma}\sum_{e\in E}\sum_{k\in [r]}v_k^e f_k(n_e({\bm\sigma}))\sum_{j\in\N:e\in\sigma_j}\beta_{ij}w_j\leq t, & \ \ \forall i\in\N\setminus\{1\}\\\vspace{0.1cm}
\displaystyle\sum_{e\in E}\sum_{k\in [r]}v_k^e f_k(n_e({\bm o}))\sum_{j\in\N:e\in o_j}\beta_{ij}w_j\leq 1, & \ \ \forall i\in\N\\\vspace{0.1cm}
v^e_k\geq 0, & \ \ \forall e\in E,k\in [r]\\
t\geq 0
\end{array}
\end{displaymath}

The dual program ${\sf DP_{CCE}}({\sf MAX},{\sf T}_{\bm w},{\bm p},{\bm o})$ is the following (again, we associate variables $y_i$, $z_i$ and $\gamma_i$ with the $i$th constraint of the first, the middle and the last family of $n$ constraints, respectively).

\begin{displaymath}
\begin{array}{ll}
minimize\displaystyle\sum_{i\in\N}\gamma_i\\\vspace{0.1cm}
subject\ to\\\vspace{0.1cm}
\displaystyle\sum_{{\bm\sigma}\in{\sf\Sigma}}p_{\bm\sigma}\sum_{i\in\N:e\in\sigma_i\setminus o_i}y_i f_k(n_e({\bm \sigma}))\sum_{j\in\N:e\in\sigma_j}\alpha_{ij}w_j\\
\ \ \displaystyle-(1+\epsilon)\sum_{{\bm\sigma}\in{\sf\Sigma}}p_{\bm\sigma}\displaystyle\sum_{i\in\N:e\in o_i\setminus\sigma_i}y_i f_k(n_e({\bm \sigma})+w_i)\left(\alpha_{ii}w_i+\sum_{j\in\N:e\in\sigma_j}\alpha_{ij}w_j\right)\\
\ \ +\displaystyle\sum_{{\bm\sigma}\in{\sf\Sigma}}p_{\bm\sigma} f_k(n_e({\bm\sigma}))\sum_{i\in\N} z_i\sum_{j\in\N:e\in \sigma_j}\beta_{ij}w_j\\
\ \ +\displaystyle f_k(n_e({\bm o}))\sum_{i\in\N}\gamma_i\sum_{j\in\N:e\in o_j}\beta_{ij}w_j\geq 0, & \ \ \forall e\in E,k\in [r]\\\vspace{0.1cm}
\displaystyle\sum_{i\in N}z_i\leq -1\\\vspace{0.1cm}
y_i,z_i,\gamma_i\geq 0, & \ \ \forall i\in\N
\end{array}
\end{displaymath}

Again, even though both ${\sf PP_{CCE}}({\sf SUM},{\sf T}_{\bm w},{\bm p},{\bm o})$ and ${\sf PP_{CCE}}({\sf MAX},{\sf T}_{\bm w},{\bm p},{\bm o})$ may be, in general, under-constrained, by the same arguments used in the discussion of the pairs of primal-dual formulations used for bounding the worst-case $\epsilon$-approximate pure price of anarchy, it follows that, for each function ${\sf SF}\in\{{\sf SUM},{\sf MAX}\}$, the optimal solution to ${\sf PP_{CCE}}({\sf SF},{\sf T}_{\bm w},{\bm p},{\bm o})$ yields an upper bound on the worst-case $\epsilon$-approximate coarse correlated price of anarchy of the class $\overline{{\cal C}}_{{\sf T}_{\bm w}}({\cal F},\alpha)$ attainable when $\bm p$ is taken for the worst $\epsilon$-approximate coarse correlated equilibrium and $\bm o$ for the social optimum (of social value $1$). Let us denote such a class with $\widehat{{\cal C}}_{{\sf T}_{\bm w}}({\cal F},\alpha)$.

The following lemma shows that any upper bound on ${\sf PPoA}_\epsilon(\beta\mbox{-}{\sf SF},{\cal C}_{{\sf T}^*_{\bm w}}({\cal F},\alpha))$ proved via the primal-dual method automatically extends to ${\sf CCPoA}_\epsilon(\beta\mbox{-}{\sf SF},\widehat{{\cal C}}_{{\sf T}_{\bm w}}({\cal F},\alpha))$.

\begin{lemma}[Extension Lemma]\label{lemma3}
For any function ${\sf SF}\in\{{\sf SUM},{\sf MAX}\}$, congestion model ${\sf T}_{\bm w}=(\N,{\bm w},E,({\sf\Sigma}_i)_{i\in\N})$, probability distribution ${\bm p}\in\Delta({\sf\Sigma}({\sf T}_{\bm w}))$ and strategy profile ${\bm o}\in{\sf\Sigma}({\sf T}_{\bm w})$, it holds that any feasible solution to ${\sf DP_{PNE}}({\sf SF},{\sf T}^*_{\bm w},{\bm\sigma},{\bm o})$ is also a feasible solution to ${\sf DP_{CCE}}({\sf SF},{\sf T}_{\bm w},{\bm p},{\bm o})$.
\end{lemma}
\begin{proof}
Let $({\bm y}^*,\gamma^*)$ be a feasible solution to ${\sf DP_{PNE}}({\sf SUM},{\sf T}^*_{\bm w},{\bm\sigma},{\bm o})$. By Property \ref{prop1} of the representative congestion model ${\sf T}^*_{\bm w}$, it follows that, for any pair of strategy profiles ${\bm\sigma},{\bm o}\in{\sf\Sigma}({\sf T}_{\bm w})$, it holds that
\begin{multline}\label{basicineq}
\sum_{i\in\N:e\in\sigma_i\setminus o_i}y^*_i f_k(n_e({\bm \sigma}))\sum_{j\in\N:e\in\sigma_j}\alpha_{ij}w_j\\
-(1+\epsilon)\sum_{i\in\N:e\in o_i\setminus\sigma_i}y^*_i f_k(n_e({\bm \sigma})+w_i)\left(\alpha_{ii}w_i+\sum_{j\in\N:e\in\sigma_j}\alpha_{ij}w_j\right)\\+\gamma^* f_k(n_e({\bm o}))\sum_{i\in\N}\sum_{j\in\N:e\in o_j}\beta_{ij}w_j\geq f_k(n_e({\bm\sigma}))\sum_{i\in\N}\sum_{j\in\N:e\in\sigma_j}\beta_{ij}w_j
\end{multline}
for any $e\in E$ and $k\in [r]$.

Since $p_{\bm\sigma}\geq 0$ for each ${\bm\sigma}\in{\sf\Sigma}({\sf T}_{\bm w})$, by multiplying inequality (\ref{basicineq}) for $p_{\bm\sigma}$ and then summing up the obtained inequalities for each ${\bm\sigma}\in{\sf\Sigma}({\sf T}_{\bm w})$, we obtain that, for each $e\in E$ and $k\in [r]$, it holds that
\begin{multline}\label{basicineq1}
\sum_{{\bm\sigma}\in{\sf\Sigma}}p_{\bm\sigma}\sum_{i\in\N:e\in\sigma_i\setminus o_i}y^*_i f_k(n_e({\bm \sigma}))\sum_{j\in\N:e\in\sigma_j}\alpha_{ij}w_j\\
-(1+\epsilon)\sum_{{\bm\sigma}\in{\sf\Sigma}}p_{\bm\sigma}\sum_{i\in\N:e\in o_i\setminus\sigma_i}y^*_i f_k(n_e({\bm \sigma})+w_i)\left(\alpha_{ii}w_i+\sum_{j\in\N:e\in\sigma_j}\alpha_{ij}w_j\right)\\+\gamma^* f_k(n_e({\bm o}))\sum_{i\in\N}\sum_{j\in\N:e\in o_j}\beta_{ij}w_j\sum_{{\bm\sigma}\in{\sf\Sigma}}p_{\bm\sigma}\geq\sum_{{\bm\sigma}\in{\sf\Sigma}}p_{\bm\sigma} f_k(n_e({\bm\sigma}))\sum_{i\in\N}\sum_{j\in\N:e\in\sigma_j}\beta_{ij}w_j.
\end{multline}
By $\sum_{{\bm\sigma}\in{\sf\Sigma}}p_{\bm\sigma}=1$, it follows that, for any $e\in E$ and $k\in [r]$, inequality $(\ref{basicineq1})$ coincides with the relative dual constraint of ${\sf DP_{CCE}}({\sf SUM},{\sf T}_{\bm w},{\bm p},{\bm o})$ and this shows that the solution $({\bm y}^*,\gamma^*)$ is also feasible for ${\sf DP_{CCE}}({\sf SUM},{\sf T}_{\bm w},{\bm p},{\bm o})$.

A similar argument shows the claim for the case of the social function $\M$.\qed
\end{proof}

\

We now have all the ingredients needed to conclude the proof of the theorem.

Fix a function ${\sf SF}\in\{{\sf SUM},{\sf MAX}\}$. Assume, first, that ${\sf PP_{PNE}}({\sf SF},{\sf T}_{\bm w}^*,{\bm \sigma}^*,{\bm o}^*)$ is unlimited. Then, by Lemma \ref{lemma2}, it holds that ${\sf PPoA}_\epsilon(\beta\mbox{-}{\sf SF},\overline{{\cal C}}_{\bm w}({\cal F},\alpha))=\infty$ which, together with ${\sf PPoA}_\epsilon(\beta\mbox{-}{\sf SF},\overline{{\cal C}}_{\bm w}({\cal F},\alpha))\leq{\sf CCPoA}_\epsilon(\beta\mbox{-}{\sf SF},\overline{{\cal C}}_{\bm w}({\cal F},\alpha))$, immediately implies that ${\sf PPoA}_\epsilon(\beta\mbox{-}{\sf SF},\overline{{\cal C}}_{\bm w}({\cal F},\alpha))={\sf CCPoA}_\epsilon(\beta\mbox{-}{\sf SF},\overline{{\cal C}}_{\bm w}({\cal F},\alpha))$. By applying Lemma \ref{normalization}, we obtain ${\sf PPoA}_\epsilon(\beta\mbox{-}{\sf SF},{\cal C}_{\bm w}({\cal F},\alpha))={\sf CCPoA}_\epsilon(\beta\mbox{-}{\sf SF},{\cal C}_{\bm w}({\cal F},\alpha))$.

In the case in which ${\sf PP_{PNE}}({\sf SF},{\sf T}_{\bm w}^*,{\bm \sigma}^*,{\bm o}^*)$ admits an optimal solution of value $\overline{x}$, by Lemma \ref{lemma2}, it holds that ${\sf PPoA}_\epsilon(\beta\mbox{-}{\sf SF},\overline{{\cal C}}_{\bm w}({\cal F},\alpha))=\overline{x}$. Moreover, by the Strong Duality Theorem, there exists a feasible solution $({\bm y}^*,\gamma^*)$ to ${\sf DP_{PNE}}({\sf SF},{\sf T}^*_{\bm w},{\bm\sigma},{\bm o})$ of value $\gamma^*=\overline{x}$. Choose an arbitrary game $({\sf CG},\alpha)\in\overline{{\cal C}}_{\bm w}({\cal F},\alpha)$ such that ${\sf CCPoA}_\epsilon(\beta\mbox{-}{\sf SF},\overline{{\cal C}}_{\bm w}({\cal F},\alpha))={\sf CCPoA}_\epsilon(\beta\mbox{-}{\sf SF},{\sf CG},\alpha)$ and let ${\sf T}_{\bm w}$ be the congestion model defining $\sf CG$, $\bm p$ be the worst $\epsilon$-approximate coarse correlated equilibrium of $({\sf CG},\alpha)$ and $\bm o$ be the social optimum (of social value $1$). By the definition of ${\sf T}_{\bm w}$, $\bm p$ and $\bm o$, it follows that the optimal solution to ${\sf PP_{CCE}}({\sf SF},{\sf T}_{\bm w},{\bm p},{\bm o})$ has a value of at least ${\sf CCPoA}_\epsilon(\beta\mbox{-}{\sf SF},\overline{{\cal C}}_{\bm w}({\cal F},\alpha))$, which, by the Weak Duality Theorem, implies in turn that any feasible solution to ${\sf DP_{CCE}}({\sf SF},{\sf T}_{\bm w},{\bm p},{\bm o})$ has a value of at least ${\sf CCPoA}_\epsilon(\beta\mbox{-}{\sf SF},\overline{{\cal C}}_{\bm w}({\cal F},\alpha))$. By Lemma \ref{lemma3}, it follows that $({\bm y}^*,\gamma^*)$ is also a feasible solution to ${\sf DP_{CCE}}({\sf SF},{\sf T}_{\bm w},{\bm p},{\bm o})$. This implies that ${\sf CCPoA}_\epsilon(\beta\mbox{-}{\sf SF},\overline{{\cal C}}_{\bm w}({\cal F},\alpha))\leq\gamma^*=\overline{x}={\sf PPoA}_\epsilon(\beta\mbox{-}{\sf SF},\overline{{\cal C}}_{\bm w}({\cal F},\alpha))$. Again, by applying Lemma \ref{normalization}, we obtain that ${\sf PPoA}_\epsilon(\beta\mbox{-}{\sf SF},{\cal C}_{\bm w}({\cal F},\alpha))={\sf CCPoA}_\epsilon(\beta\mbox{-}{\sf SF},{\cal C}_{\bm w}({\cal F},\alpha))$.

It is clear from our discussion that the value ${\sf PPoA}_\epsilon(\beta\mbox{-}{\sf SF},\overline{{\cal C}}_{\bm w}({\cal F},\alpha))={\sf PPoA}_\epsilon(\beta\mbox{-}{\sf SF},{\cal C}_{\bm w}({\cal F},\alpha))$ can always be (theoretically) determined via the primal-dual method, that is, by computing the value of the optimal solution of either the primal program ${\sf PP_{PNE}}({\sf SF},{\sf T}^*_{\bm w},{\bm\sigma},{\bm o})$ or the dual one ${\sf DP_{PNE}}({\sf SF},{\sf T}^*_{\bm w},{\bm\sigma},{\bm o})$ for each function ${\sf SF}\in\{{\sf SUM},{\sf MAX}\}$, and this concludes the proof (solving the dual program, in particular, requires to determine the minimum value $\gamma^*$ for which all the $r\cdot 4^n$ possible constraints induced by the $|E^*|=4^n$ pairs of values yielded by the representative congestion model ${\sf T}^*_{\bm w}$ on each of the $r$ components of the latency functions are satisfied).\qed
\end{proof}

\section{Conclusions and Open Problems}

By introducing the notions of smooth games and robust price of anarchy, Roughgarden \cite{R09,R12} showed that the class of congestion games with non-negative and non-decreasing latency functions is tight under the utilitarian social function (see Section \ref{sec-pd} of the Appendix for formal definitions). This result has been extended to the class of weighted congestion games by Bhawalkar, Gairing and Roughgarden \cite{BGR10}.
By exploiting the primal-dual method we introduced in \cite{B12}, we have generalized this result along four directions. In fact, our tightness result holds for the class of generalized weighted congestion games, for generalizations of both the utilitarian and the egalitarian social functions, for any non-negative (and possibly decreasing) latency functions and for the approximated version of the price of anarchy.

The fact that two different and seemingly uncorrelated approaches may produce the same type of general results is quite interesting. Understanding whether there is some kind of relationships between them is an intriguing question. Both approaches set up some machinery (smoothness argument vs. primal-dual formulation) allowing for the proof of significant upper bounds on the pure price of anarchy of the games under analysis and then make use of an extension theorem to show that such bounds extend to the coarse correlated price of anarchy as well.

In particular, how is this last step achieved?

Note that the proof of the smoothness argument (see the proof of Lemma \ref{smootharg}) requires the definition of smoothness to hold only for any pair of strategy profiles $({\bm\sigma},{\bm\sigma}')$ such that $\bm\sigma$ is a pure Nash equilibrium and ${\bm\sigma}'$ is a social optimum. The reason why the definition of smoothness is extended to encompass all possible pairs of strategy profiles is due to the fact that it is indeed the proof of the extension theorem (see the proof of Theorem \ref{extth}) that asks for such a stronger hypothesis. Finally, being a coarse correlated equilibrium a particular probability distribution defined over the set of strategy profiles, the notion of smoothness characterizing each profile in the support of any such an equilibrium can be suitably exploited by the linearity of expectation.

In the primal-dual method, instead, the higher degree of generality that is needed to move from pure Nash equilibria up to coarse correlated equilibria is provided by the representative congestion model which imposes that the variables yielding a feasible solution to the dual formulation for the pure price of anarchy have to satisfy any type of ``pure dual constraint'' that may eventually arise by considering all possible types of configurations of congestions. Then, since it turns out that the dual constraint characterizing the dual formulation for the coarse correlated price of anarchy is indeed a convex combination of a subset of all the possible ``pure dual constraints'', the extension theorem follows immediately.

Anyway, there is an major difference between the two methods when one aims at showing the tightness of a particular class of games.
When adopting the smoothness framework, after having proved that a class of games is $(\lambda,\mu)$-smooth for a certain pair of parameters $\lambda$ and $\mu$, one has to show that there exists a game in the class for which the pure price of anarchy is indeed $\frac{\lambda}{1-\mu}$, that is, that the price of anarchy of pure Nash equilibria matches the robust price of anarchy. We stress that this step can be avoided when adopting the primal-dual formulation, since it is directly implied by the Duality Theory (see Lemma \ref{lemma2}). In fact, note that the notion of robust price of anarchy, as the best possible upper bound on the pure price of anarchy achievable via the smoothness argument, has no correspondent in the primal-dual method where this bound is implicitly defined by the optimal solution of the pair of primal-dual programs.

By summarizing, our findings seem to reveal that the primal-dual method may be superior to the smoothness framework within the realm of weighted congestion games and their possible generalizations, but, at the same time, the primal-dual method has never been exploited so far outside this realm. Hence, a good starting point would be that of trying to export it to other scenarios of investigation in which the smoothness framework has already been fruitfully applied, such as, for instance, the quantification of the price of anarchy in unrelated scheduling games, valid utility games, opinion formation games and auction theory.

\section*{Appendix}

\section{The Smoothness Argument and the Robust Price of Anarchy}\label{sec-pd}

Let ${\cal G}=\left(\N,({\sf\Sigma}_i)_{i\in\N},(c_i)_{i\in\N}\right)$ be a cost minimization game defined by the set of players $\N$, the set of strategies ${\sf\Sigma}_i$ and the individual cost function $c_i:{\sf\Sigma}\rightarrow\RPP$ for each player $i\in\N$. A social function ${\sf SF}:{\sf\Sigma}\rightarrow\RPP$ for $\cal G$ is sum-bounded if, for each ${\bm\sigma}\in\sf\Sigma$, it holds that ${\sf SF}({\bm\sigma})\leq\sum_{i\in\N}c_i({\bm\sigma})$.

\begin{definition}[Smoothness]
Given a social function $\sf SF$, $\cal G$ is $(\lambda,\mu)$-smooth under $\sf SF$ if, for any two strategy profiles ${\bm\sigma},{\bm\sigma}'\in\sf\Sigma$,
it holds that $\sum_{i\in\N}c_i({\bm\sigma}_{-i},\sigma'_i)\leq\lambda{\sf SF}({\bm\sigma}')+\mu{\sf SF}({\bm\sigma})$.
\end{definition}

The connection between the notion of smoothness and that of pure price of anarchy is captured by the following lemma.

\begin{lemma}[Smoothness Argument]\label{smootharg}
If $\cal G$ is $(\lambda,\mu)$-smooth under a sum-bounded social function $\sf SF$, with $\lambda>0$ and $\mu<1$, then it holds that ${\sf PPoA}({\sf SF},{\cal G})\leq\frac{\lambda}{1-\mu}$.
\end{lemma}
\begin{proof}
Let $\bm\sigma$ be any pure Nash equilibrium for $\cal G$ and $\bm o$ be a social optimum for $\cal G$ under $\sf SF$. It holds that
$${\sf SF}({\bm\sigma})\leq\sum_{i\in\N}c_i({\bm\sigma})\leq\sum_{i\in\N}c_i({\bm\sigma}_{-i},o_i)\leq\lambda{\sf SF}({\bm o})+\mu{\sf SF}({\bm\sigma})$$
and the claim follows by rearranging the terms.\qed
\end{proof}

\

The robust price of anarchy is then defined as the best possible upper bound on the pure price of anarchy that can be proved via the smoothness argument. For a game $\cal G$ and a social function $\sf SF$, we denote with ${\cal A}_{\sf SF}({\cal G})$ the set of parameters $(\lambda,\mu)$ such that $\cal G$ is $(\lambda,\mu)$-smooth under $\sf SF$.

\begin{definition}[Robust Price of Anarchy]
Given a sum-bounded social function $\sf SF$, the robust price of anarchy of $\cal G$ under $\sf SF$, is the value $\rho_{{\sf SF}}({\cal G})=\textrm{inf}\left\{\frac{\lambda}{1-\mu}:(\lambda,\mu)\in{\cal A}_{\sf SF}({\cal G})\right\}$.
\end{definition}

The power of the smoothness argument is then stressed by the following extension theorem.

\begin{theorem}[Extension Theorem]\label{extth}
For each cost minimization game $\cal G$ and sum-bounded social function $\sf SF$ for $\cal G$, it holds that ${\sf CCPoA}({\sf SF},{\cal G})\leq\rho_{\sf SF}({\cal G})$.
\end{theorem}
\begin{proof}
Let $\bm p$ be any coarse correlated equilibrium for $\cal G$ and $\bm o$ be a social optimum for $\cal G$ under $\sf SF$. It holds that
$$\sum_{{\bm\sigma}\in{\sf\Sigma}}p_{\bm\sigma}\cdot{\sf SF}({\bm\sigma})\leq\sum_{{\bm\sigma}\in{\sf\Sigma}}p_{\bm\sigma}\sum_{i\in\N}c_i({\bm\sigma})
=\sum_{i\in\N}\sum_{{\bm\sigma}\in{\sf\Sigma}}p_{\bm\sigma}\cdot c_i({\bm\sigma})
\leq\sum_{i\in\N}\sum_{{\bm\sigma}\in{\sf\Sigma}}p_{\bm\sigma}\cdot c_i({\bm\sigma}_{-i},o_i)$$
$$=\sum_{{\bm\sigma}\in{\sf\Sigma}}p_{\bm\sigma}\sum_{i\in\N}c_i({\bm\sigma}_{-i},o_i)\leq\sum_{{\bm\sigma}\in{\sf\Sigma}}p_{\bm\sigma}\left(\lambda{\sf SF}({\bm o})+\mu{\sf SF}({\bm\sigma})\right)=\mu\sum_{{\bm\sigma}\in{\sf\Sigma}}p_{\bm\sigma}\cdot{\sf SF}({\bm\sigma})+\lambda{\sf SF}({\bm o})$$
and the claim follows by rearranging the terms.\qed
\end{proof}

\

Let $\cal C$ be a class of cost minimization games and $\widehat{{\cal C}}\subseteq{\cal C}$ be the subclass of the games in $\cal C$ which admit at least one pure Nash equilibrium. Given a social function $\sf SF$, we denote with ${\cal A}_{\sf SF}({\cal C})$ the set of parameters $(\lambda,\mu)$ such that each game ${\cal G}\in\widehat{{\cal C}}$ is $(\lambda,\mu)$-smooth under $\sf SF$.

\begin{definition}[Tight Class of Games]
A class of games $\cal C$ is tight under the social function $\sf SF$ if it holds that $\textrm{sup}_{{\cal G}\in\widehat{{\cal C}}}{\sf PPoA}({\sf SF},{\cal G})=\textrm{inf}_{(\lambda,\mu)\in{\cal A}({\cal C})}\frac{\lambda}{1-\mu}$.
\end{definition}

\end{document}